\documentclass[oribibl]{llncs}
\usepackage{makeidx}  
\usepackage{amsmath,amssymb,amsfonts}
\usepackage{color}
\usepackage{enumerate}
\usepackage{mathtools}
\usepackage{algorithm}
\usepackage[noend]{algpseudocode}
\usepackage{tikz}
\usepackage{thm-restate}
\usepackage{hyperref}
\usepackage{cleveref}
\usepackage{soul}

\hypersetup{
    colorlinks,
    linkcolor={red},
    citecolor={blue},
    urlcolor={blue}
}

\usepackage{graphicx}
\usepackage{caption}
\usepackage{subfigure}

\usepackage[mathcal]{euscript}
\usepackage{mathptmx}

\usepackage{color}

\newcommand{\lca}{\ensuremath{\operatorname{lca}}}
\newcommand{\Virr}{\ensuremath{V^{\times}_{\mathrm{irr}}}}
\newcommand{\Gs}{\ensuremath{G^{\star}}}
\newcommand{\Es}{\ensuremath{E^{\star}}}
\newcommand{\Rs}{\ensuremath{R^{\star}}}
\newcommand{\Hs}{\ensuremath{\mathcal{H}^{\star}}}

\newcommand{\Rone}{\ensuremath{\overrightarrow{1}}}

\begin{document}
\frontmatter          
\pagestyle{headings}  
\mainmatter              
\title{Partial Homology Relations \-- Satisfiability in terms of Di-Cographs}
\titlerunning{XXX}  
%
\author{Nikolai N{\o}jgaard\inst{1,2} \and Nadia El-Mabrouk\inst{3} \and Daniel Merkle\inst{2} Nicolas Wieseke\inst{4} \and Marc Hellmuth\inst{1,5}}

\authorrunning{N{\o}jgaard et al.} 
%
%
\institute{
Institute of Mathematics and Computer Science, University of Greifswald, Walther-
  Rathenau-Strasse 47, D-17487 Greifswald, Germany 
\email{mhellmuth@mailbox.org}
\and
Department of Mathematics and Computer Science,
				    University of Southern Denmark, 
						Campusvej 55,	DK-5230 Odense M, Denmark \\
\and
	Department of Computer Science
	and Operational Research,
	University of Montreal, 
	CP 6128 succ Centre-Ville,
	Montreal,
	Canada \\
\and
 Swarm Intelligence and Complex Systems Group,
        Department of Computer Science, University of Leipzig,
        Augustusplatz 10, D-04109 Leipzig, Germany\\
\and
Center for Bioinformatics, Saarland University, Building E 2.1, D-66041
Saarbr{\"u}cken, Germany \\
}

\maketitle              

\begin{abstract}
Directed cographs (di-cographs) play a crucial role in the reconstruction of
	evolutionary histories of genes based on homology relations which
are binary relations between genes. A variety of methods based on pairwise
	sequence comparisons can be used to infer such homology relations (e.g.\
	orthology, paralogy, xenology).  
	They are \emph{satisfiable} if the relations can be explained by an event-labeled gene tree, i.e., 
they can simultaneously co-exist in an evolutionary history of the underlying genes. 
Every gene tree is equivalently interpreted as a so-called cotree that entirely encodes the structure of a di-cograph. 
Thus, satisfiable homology relations must necessarily form a di-cograph. 
The inferred homology relations might not cover each pair of genes and thus,  provide only partial knowledge 
on the full set of homology relations. Moreover, for particular pairs of genes, it might be known with a 
high degree of certainty that they are not orthologs (resp.\ paralogs, xenologs) which yields forbidden pairs of genes.  
Motivated by this observation, we characterize (partial) satisfiable homology relations  with or without forbidden gene pairs, 
provide a quadratic-time algorithm for their recognition and for the
	computation of a cotree that explains the given relations. 

\keywords{Directed Cographs, Partial Relations, Forbidden Relations, Recognition Algorithm, Homology, Orthology, Paralogy, Xenology}
\end{abstract}

\sloppy
\section{Introduction}\vspace{-0.1in}
Directed cographs (di-cographs) are a well-studied class of graphs that 
can uniquely be represented by an ordered rooted tree $(T,t)$, called cotree,  
where each inner vertex gets a
unique label ``$1$'', ``$\Rone$'' or ``$0$'' \cite{BLS:99,Corneil:81,CP-06,EHPR:96}. 
In particular, di-cographs have been shown to play an important role for the reconstruction of
the evolutionary history of genes
or species based on genomic sequence data
\cite{HSW:16,HHH+13,Hellmuth:15a,lafond2015orthology,DEML:16,LDEM:16}. Genes are
the molecular units of heredity holding the information to build and maintain
cells. During evolution, they are mutated, duplicated, lost and passed to
organisms through speciation or horizontal gene transfer (HGT), which is the
exchange of genetic material among co-existing species. A gene family comprises
genes sharing a common origin. Genes within a gene family are called
\emph{homologs}. 

The history of a gene family is equivalent to an event-labeled
gene tree, the leaves correspond to extant genes, internal vertices to ancestral genes
and the label of an internal vertex highlighting the event at the origin of the
divergence leading to the offspring, namely \emph{speciation-},
\emph{duplication-} or \emph{HGT-events} \cite{Fitch2000}. Equivalently, the
history of genes is described by an event-labeled rooted tree for which each
inner vertex gets a unique label ``$1$'' (for speciation) , ``$0$''
(for duplication) or ``$\Rone$'' (for HGT). In other words, any
gene tree is also a cotree. The type of event ``$1$'', ``$0$'' and ``$\Rone$''
of the lowest common ancestor of two genes gives rise to one of three distinct
\emph{homology relations} respectively, the \textit{orthology-relation} $R_1$, the
\textit{paralogy-relation}
$R_0$ and the \textit{xenology-relation} $R_{\Rone}$. The
orthology-relation $R_1$ on a set of genes forms an undirected cograph
\cite{HHH+13}. In \cite{HSW:16} it has been shown that the graph $G$ with arc
set $E(G) = R_1\cup R_{\Rone}$ must be a di-cograph (see \cite{HW:16b} for a
detailed discussion). 

In practice, these homology relations are often estimated
from sequence similarities and synteny information, without requiring any
\emph{a priori} knowledge on the topology of either the gene tree or the species
tree (see e.g.\
\cite{Altenhoff:09,AGGD:13,ASGD:11,CMSR:06,Lechner:11a,Lechner:14,inparanoid:10,TG+00,T+11,Dessimoz2008,LH:92,RSLD15}).
The starting point of this contribution is a set of estimated relations $R_1$, $R_{\Rone}$
and $R_0$. In particular, we consider so-called \emph{partial} and
\emph{forbidden} relations: In fact,
similarity-based homology inference methods often depend on particular threshold
parameters to determine whether a given pair of genes is in one of $R_0,R_1$ or
$R_{\Rone}$. Gene pairs whose sequence similarity falls below (or above) a
certain threshold cannot be unambiguously identified as belonging to one of the
considered homology relations. Hence, in practice one usually obtains partial
relations only, as only subsets of these relations may
be known. Moreover, different homology inference methods may lead to different
predictions. Thus, instead of a yes or no orthology, paralogy or xenology assignment,
a confidence score can rather be assigned to each relation~\cite{DEML:16}. A
simple way of handling such weighted gene pairs is to set an upper threshold
above which a relation is predicted, and a lower threshold under which a
relation is rejected, leading to partial relations with forbidden gene pairs,
i.e., gene pairs that are not included in any of the three relations but for
which it is additionally known to which relations they definitely \emph{not}
belong to.

In this contribution, we generalize results established by 
Lafond and El-Mabrouk \cite{Lafond2014} and
characterize satisfiable partial relations with and
without forbidden relations. We provide a recursive quadratic-time algorithm 
testing whether the considered relations are satisfiable, and if so
reconstructing a corresponding cotree. This, in turn, allows us to extend
satisfiable partial relations to full relations. Finally, we
evaluate the accuracy of the designed algorithm on large-scaled simulated data
sets. As it turns out, it suffices to have only a very few but correct pairs
of relations to recover most of them.

Note, the results established here may also be of interest for a broader 
scientific community. Di-cographs play an important role in computer science 
because many combinatorial  optimization problems that are NP-complete for arbitrary graphs 
become polynomial-time solvable on di-cographs \cite{Corneil:85, BLS:99,Gao:13}.
However, the cograph-editing problem is NP-hard \cite{Liu:12}. 
Thus, an attractive starting point for heuristics that edit a given graph to
a cograph  may be the knowledge of satisfiable parts that 
eventually lead to partial information of the underlying di-cograph structure
of the graph of interest.

\vspace{-0.1in}
\section{Preliminaries}\vspace{-0.1in}
\label{sec:prelim}

\paragraph{\bf Basics.}
In what follows, we always consider \emph{binary} and \emph{irreflexive}
relations $R\subseteq \Virr\coloneqq V\times V \setminus \{(v,v)\mid v\in V\}$
and we omit to mention it each time. If we have a non-symmetric relation $R$,
then we denote $R^{\mathrm{sym}} = R \cup \{(x,y) \mid (y,x)\in R\}$ the
\emph{symmetric extension} of $R$ and by $\overleftarrow{R}$ the set $\{(x,y)
\mid (y,x)\in R\}$. For a subset $W\subseteq V$ and a relation $R$, we define
$R[W] = \{(x,y)\in R \mid x,y\in W\}$ as the \emph{sub-relation of $R$ that is
induced by $W$}. Moreover, for a set of relations $\mathcal{R}=\{R_1,\dots,
R_n\}$ we set $\mathcal{R}[W]=\{R_1[W],\dots, R_n[W]\}$.

A directed graph (digraph) $G=(V,E)$ has vertex set $V(G)=V$ and arc set
$E(G)=E\subseteq \Virr$. Given two disjoint digraphs $G=(V,E)$ and $H=(W,F)$,
the digraph $G\cup H = (V\cup W, E\cup F)$, $G\oplus H = (V\cup W, E\cup F\cup
\{(x,y),(y,x)\mid x\in V,y\in W\})$ and $G\oslash H = (V\cup W, E\cup F\cup
\{(x,y)\mid x\in V,y\in W\})$ denote the \emph{union}, \emph{join} and
\emph{directed join} of $G$ and $H$, respectively. For a given subset $W\subseteq V$, 
the \emph{induced subgraph} $G[W]=(W,F)$ of $G=(V,E)$ is the subgraph for which $x,y\in W$ and $(x,y)\in
E$ implies that $(x,y)\in F$. We call $W\subseteq V$ a \emph{(strongly)
connected component} of $G=(V,E)$ if $G[W]$ is a \emph{maximal} (strongly)
connected subgraph of $G$. 

Given a digraph $G=(V,E)$ and a partition $\{V_1,V_2,\dots,V_k\}$ of its vertex
set $V$, the \emph{quotient digraph $G/\{V_1,V_2,\dots,V_k\}$ has as vertex
set} $\{V_1,V_2,\dots,V_k\}$ and two distinct vertices $V_i, V_j$ form an arc
$(V_i, V_j)$ in $G/\{V_1,\dots,V_k\}$ if there are vertices $x\in V_i$ and $y\in
V_j$ with $(x,y)\in E$.

An acyclic digraph is called \emph{DAG}. It is well-known that the vertices of
a DAG can be \emph{topologically ordered}, i.e., there is an ordering of the
vertices as $v_1, \dots, v_n$ such that $(v_i,v_j)\in E$ implies that $i<j$. To
check whether a digraph $G$ contains no cycles one can equivalently check
whether there is a topological order, which can be done via a depth-first search
in $O(|V(G)|+|E(G)|)$ time.


Furthermore, we consider a \emph{rooted tree $T=(W,E)$ (on $V$)} with root
$\rho_T \in V$ and leaf set $V\subseteq W$ such that the root has degree $ \geq
2$ and each vertex $v\in W\setminus V$ with $v\neq \rho_T$ has degree $ \geq 3$.
We write $x \succeq_T y$, if $x$ lies on the path from $\rho_T$ to $y$. The
\emph{children} of a vertex $x$ are all adjacent vertices $y$ for which $x \succeq_T
y$. Given two leaves $x,y\in V$, their lowest common ancestor $\lca(x,y)$ is the
first vertex that lies on both paths from $x$ to the root and $y$ to the root.
We say that rooted trees $T_1, \dots, T_k$, $k\geq 2$ \emph{are joined under a
new root in the tree $T$} if $T$ is obtained by the following procedure: add a
new root $\rho_T$ and all trees $T_1, \dots, T_k$ to $T$ and connect the root
$\rho_{T_i}$ of each tree $T_i$ to $\rho_T$ with an edge $(\rho_T,\rho_{T_i})$.

\paragraph{\bf Di-cographs.}

Di-cographs generalize the notion of undirected cographs \cite{EHPR:96,CP-06,Corneil:81,BLS:99} and 
are defined  recursively as follows: The single vertex
graph $K_1$ is a di-cograph, and if $G$ and $H$ are di-cographs, then $G\cup H$,
$G\oplus H$, and $G\oslash H$ are di-cographs
\cite{gurski2017dynamic,Corneil:81}. 
Each Di-cograph $G=(V,E)$ is associated with a unique ordered least-resolved tree $T=(W,F)$ (called \emph{cotree}) 
with leaf set $L=V$ and a 
labeling function $t:W\setminus L\to \{0,1,\overrightarrow{1}\}$ defined by\vspace{-0.1in}
\begin{align*}
  t(\lca(x,y)) = \begin{cases}
    0,    & \text{ if } (x,y),(y,x)\notin E \\
    1,     & \text{ if } (x,y),(y,x)\in E   \\
    \Rone, &\text{ otherwise.} 
 \end{cases}
\end{align*}
Since the vertices in the cotree $T$ are
ordered, the label $\overrightarrow{1}$ on some $\lca(x,y)$ of two distinct
leaves $x,y\in L$ means that there is an arc $(x,y) \in E$, while $(y,x)
\notin E$, whenever $x$ is placed to the left of $y$ in $T$.

Some important properties of di-cographs that we need for later reference are given now. 
\begin{lemma}[\cite{Corneil:81,gurski2017dynamic,mcconnell2005linear}]
		A digraph $G=(V,E)$ is a di-cograph if and only if each induced subgraph of $G$ is a di-cograph.

		Determining whether a digraph is a di-cograph, and if so, computing the corresponding
		cotree can be done in $O(|V|+|E|)$ time.
	 	\label{lem:binary-induced}
\end{lemma}

\vspace{-0.1in}
\section{Problem Statement}\vspace{-0.1in}

As argued in the introduction and explained in more detail in \cite{HSW:16}, 
the evolutionary history of genes is equivalently described by an ordered 
rooted tree $\mathcal{T}=({T},{t})$ where the leaf set of
${T}$ are genes and $t$ is a map that  
assigns to each non-leaf vertex a unique label $0,1$ or $\Rone$. The labels correspond to the classical
evolutionary events that act on the genes through evolution, namely
\emph{speciation} ($1$), \emph{duplication} ($0$) and \emph{horizontal gene
transfer (HGT)} ($\Rone$). 
The tree $T$ is ordered to represent the inherently asymmetric nature of HGT
events with their unambiguous distinction between the vertically transmitted
``original'' gene and the horizontally transmitted ``copy''. 

Therefore, a given gene tree $\mathcal{T}=({T},{t})$ 
is a cotree of some di-cograph $G(\mathcal{T})=(V,E)$. 
In particular, $\mathcal{T}$ gives rise to 
the following three well-known \emph{homology relations} between genes:\vspace{-0.1in}
\begin{align*}
	&\text{the \emph{orthology-relation}:} &&R_{1}({\mathcal{T}}) = \{(x,y)\mid t(\lca(x,y)) =1\}, \\ 
	&\text{the \emph{paralogy-relation}:}  &&R_{0}({\mathcal{T}}) = \{(x,y)\mid  t(\lca(x,y)) =0\}, \text{ and}\\ 
&\text{the \emph{xenology-relation}: }
					&&R_{\Rone}({\mathcal{T}}) = \{(x,y)\mid t(\lca(x,y))=\Rone \text{ and } x \text{ is left of } y \text{ in } T\}. 
\end{align*} 
Equivalently, 
$R_{1}({\mathcal{T}}) = \{(x,y)\mid (x,y),(y,x)\in E\}$, 
$R_{0}({\mathcal{T}}) = \{(x,y)\mid  (x,y),(y,x)\notin E\}$, 
					$R_{\Rone}({\mathcal{T}}) = \{(x,y)\mid  (x,y)\in E, (y,x)\notin E\}$.
By construction, 
$R_{1}({\mathcal{T}})$ and $R_{0}({\mathcal{T}})$ are symmetric relations, 
while $R_{\Rone}({\mathcal{T}})$ is an anti-symmetric relation. 

In practice, however, one often has only empirical estimates $R_{0},R_{1}$ and $R_{\Rone}$
of some ``true'' relations $R_{0}({\mathcal{T}}),R_{1}({\mathcal{T}})$ and $R_{\Rone}({\mathcal{T}})$, respectively.
Moreover, it is often the case that for two distinct leaves 
$x,y$ none of the pairs $(x,y)$ and $(y,x)$ is contained in the
estimates $R_{0},R_{1}$ and $R_{\Rone}$. 

In what follows we always assume that $R_{0},R_{1}$ and $R_{\Rone}$ are subsets of $\Virr$
and pairwisely disjoint. Furthermore $R_0$ and $R_1$ are always 
symmetric relations while $R_{\Rone}$ is always an anti-symmetric relation.

\begin{definition}[Full and Partial Relations]
A set $\mathcal{H} = \{R_{0},R_{1},R_{\Rone}\}$ of relations  is
\emph{full} if $R_0\cup R_1\cup R_{\Rone}^{\mathrm{sym}} = \Virr$
and \emph{partial} if $R_0\cup R_1\cup R_{\Rone}^{\mathrm{sym}} \subseteq  \Virr$. 
\end{definition}
The definition allows considering 
full relations as partial. In other words, all results that will be obtained
for partial relations will also be valid for full relations. 

The question arises under which conditions the given partial relations $R_{0},R_{1}$ and $R_{\Rone}$
are satisfiable, i.e., there is a cotree $\mathcal{T}=(T,t)$ such that 
$R_{1}\subseteq R_{1}({\mathcal{T}})$, $R_{0}\subseteq R_{0}({\mathcal{T}})$ and
$R_{\Rone} \subseteq R_{\Rone}({\mathcal{T}})$, or 
equivalently, there is a di-cograph $\Gs=(V,\Es)$ such that $(R_1\cup R_{\Rone})\subseteq \Es$
and $R_0\cap \Es = \overleftarrow{R_{\Rone}} \cap \Es = \emptyset$.

\begin{definition}[Satisfiable Relations]
A full set $\mathcal{H} = \{R_{0},R_{1},R_{\Rone}\}$ is \emph{satisfiable}, if there is a cotree
$\mathcal{T} = ({T},{t})$ such that $R_{1}=R_{1}({\mathcal{T}})$, $R_{0}=R_{0}({\mathcal{T}})$ and
$R_{\Rone} = R_{\Rone}({\mathcal{T}})$. 

	A partial set $\mathcal{H} = \{R_{0},R_{1},R_{\Rone}\}$ is satisfiable, 
	if there is a full set $\Hs = \{\Rs_{0},\Rs_{1},\Rs_{\Rone}\}$
	that is satisfiable such that $R_0\subseteq \Rs_0$, $R_1\subseteq \Rs_1$, and $R_{\Rone}\subseteq \Rs_{\Rone}$. 

	In this case, we say that $\mathcal{H}$ can be extended to a satisfiable full set $\Hs$
	and that $\mathcal{T}$ \emph{explains} $\mathcal{H}$ and $\Hs$.
\end{definition}

It is easy to see that a full set $\mathcal{H}$ is satisfiable if and only if 
the graph $G(R_{1},R_{\Rone}) = (V,R_1\cup R_{\Rone})$ is a di-cograph. The latter result has already been observed in 
\cite{HSW:16} and is summarized bellow.

\begin{theorem}[\cite{HSW:16}]
	  The full set $\mathcal{H} = \{R_{0},R_{1},R_{\Rone}\}$ is satisfiable if and only
		if $G(R_{1},R_{\Rone}) = (V,R_1\cup R_{\Rone})$ is a di-cograph.
\label{thm:character-satisf}
\end{theorem}

Due to errors and noise in the data, the graph
$G(R_{1},R_{\Rone})$ is often not a di-cograph. However, in case that 
 $\mathcal{H}$ is partial, it might be possible to extend  
$G(R_{1},R_{\Rone})$ to a di-cograph. 
Moreover, in practice, one often has  additional knowledge about the unknown parts of the relations,
that is, one may know that a pair $(x,y)$ is \emph{not} in
relation $R_i$ for some $i \in \{0,1,\Rone\}$. To be able to model such
forbidden pairs,  we introduce the concept of satisfiability in terms of 
forbidden relations.
\begin{definition}[Satisfiable Relations w.r.t.\ Forbidden Relations] 
	Let $\mathcal{F} = \{F_0, F_1,F_{\Rone}\}$ be a partial set of relations. 
	We say that a full set $\mathcal{H} = \{R_{0},R_{1},R_{\Rone}\}$ is satisfiable w.r.t. $\mathcal{F}$
	if  $\mathcal{H}$ is satisfiable and 
	\[R_0\cap F_0 = R_1\cap F_1 =  R_{\Rone}\cap F_{\Rone} = \emptyset. \] 
	On the other hand, a partial set  $\mathcal{H} = \{R_{0},R_{1},R_{\Rone}\}$ is satisfiable w.r.t. $\mathcal{F}$
	if $\mathcal{H}$ can be extended to a full set $\Hs$ that is satisfiable w.r.t. $\mathcal{F}$.
\end{definition}
Equivalently, a partial set  $\mathcal{H} = \{R_{0},R_{1},R_{\Rone}\}$ is satisfiable w.r.t. $\mathcal{F}$, if
there is a cotree  $\mathcal{T}=({T},{t})$ such that 
$R_0\subseteq R_0(\mathcal{T})$, $R_1\subseteq R_1(\mathcal{T})$,
$R_{\Rone}\subseteq R_{\Rone}(\mathcal{T})$ and 
$R_0(\mathcal{T})\cap F_0 = R_1(\mathcal{T})\cap F_1 =  R(\mathcal{T})_{\Rone}\cap F_{\Rone} = \emptyset$.

\vspace{-0.1in}
\section{Satisfiable Relations}\vspace{-0.1in}
\label{sec:characterize}

In what follows, we consider the problem of deciding whether a partial set $\mathcal{H}$ of relations 
is satisfiable, and if so, finding an extended full set $\Hs$ of $\mathcal{H}$ and the respective cotree
that explains $\mathcal{H}$. 
\emph{Due to space limitation, all proofs are omitted and can be found in the appendix.}

Based on results provided by B\"ocker and Dress \cite{Boeckner:98}, 
the following result has been established for the HGT-free case.
\begin{theorem}[\cite{Lafond2014,HHH+13}] 
Let  $R_{\Rone}=\emptyset$ , $F_0=F_1=\emptyset$ and $F_{\Rone}=\Virr$.
A full set $\mathcal{H} = \{R_{0},R_{1}, R_{\Rone}\}$  is
satisfiable w.r.t. $\mathcal{F} = \{F_0,F_1,F_{\Rone}\}$ if and only if the graph $G=(V,R_1)$ is an undirected cograph.

A partial set  $\mathcal{H} = \{R_{0},R_{1},\emptyset\}$ is satisfiable w.r.t.\ $\mathcal{F}$ if and only if at least one of the
following statements is satisfied:\vspace{-0.1in}
\begin{enumerate}
\item $G=(V,R_1)$ is disconnected and each of its connected components is satisfiable.
\item $G=(V,R_0)$ is disconnected and each of its connected components is satisfiable.
\end{enumerate}
\end{theorem}
To generalize the latter result for non-empty relations $R_{\Rone}, F_0$ and $F_1$
and to allow pairs to be added to $R_{\Rone}$, i.e., $F_{\Rone}\neq\Virr$,
we need the following result.
\begin{lemma}\label{lem:satisf-partition}
	A partial set $\mathcal{H} = \{ R_0,R_1, R_{\Rone}\}$ is satisfiable w.r.t.\ the
	set of forbidden relations $\mathcal{F} = \{ F_1, F_0, F_{\Rone} \}$ 
	if and only if for any partition 	$\{C_1, \dots,
	C_k\}$ of $V$ the set 
	$\mathcal{H}[C_i]$ is satisfiable w.r.t.\ $\mathcal{F}[C_i]$, $1\leq i\leq k$.
\end{lemma}
Lemma \ref{lem:satisf-partition} characterizes satisfiable partial sets $\mathcal{H} =\{ R_0,
R_1, R_{\Rone}\}$ in terms of a partition $\{C_1, \dots, C_k\}$ of $V$ and the
induced sub-relations in  $\mathcal{H}[C_i]$. 
In what follows, we say that \emph{a component} $C$ of a graph \emph{ is satisfiable}
if the set $\mathcal{H}[C]$ is satisfiable w.r.t.\ $\mathcal{F}[C]$.

We are now able to state the main result. 
\begin{theorem} \label{thm:sat}
Let $\mathcal{H} = \{ R_0,R_1, R_{\Rone}\}$ be a partial set. 
Then, $\mathcal{H}$ is satisfiable w.r.t.\  the set of forbidden relations $\mathcal{F} =  \{ F_0,
	F_1,F_{\Rone} \}$ if and only if 
	$R_0 \cap F_0 = R_1 \cap F_1 = R_{\Rone} \cap F_{\Rone} = \emptyset$ and at least one of the following statements hold: \vspace{-0.1in}
\begin{description}
	\item[] \emph{Rule (0):\ } $|V| = 1$. 
	\item[] \emph{Rule (1):\ }
		(a) $G_0 \coloneqq (V,R_1\cup R_{\Rone}\cup F_0)$ is disconnected and (b) each connected component of $G_0$ is satisfiable. 
	\item[] \emph{Rule (2):\ } (a) $G_1 \coloneqq (V,R_0\cup R_{\Rone}\cup F_{1})$ is disconnected and (b) each connected component of $G_1$ is satisfiable.  
	\item[] \emph{Rule (3):\ } 
	(a) $G_{\Rone} \coloneqq (V,R_0\cup R_1\cup R_{\Rone} \cup \overleftarrow{F_{\Rone}})$ 
		  contains more than one strongly connected
			component, and (b) each strongly connected component of $G_{\Rone}$ is satisfiable. 
\end{description}
\end{theorem}

Note, the notation $G_0$, $G_1$ and $G_{\Rone}$ in Thm.\ \ref{thm:sat}
is chosen because if $G_0$ (resp.\ $G_1$ and $G_{\Rone}$) satisfies Rule (1) (resp.\ (2) and (3)), 
then the root of the cotree that explains $\mathcal{H}$ is labeled ``$0$''
(resp.\ ``$1$'' and ``$\Rone$'').

\begin{figure}[t]
\centering     
	\includegraphics[width=.8\textwidth]{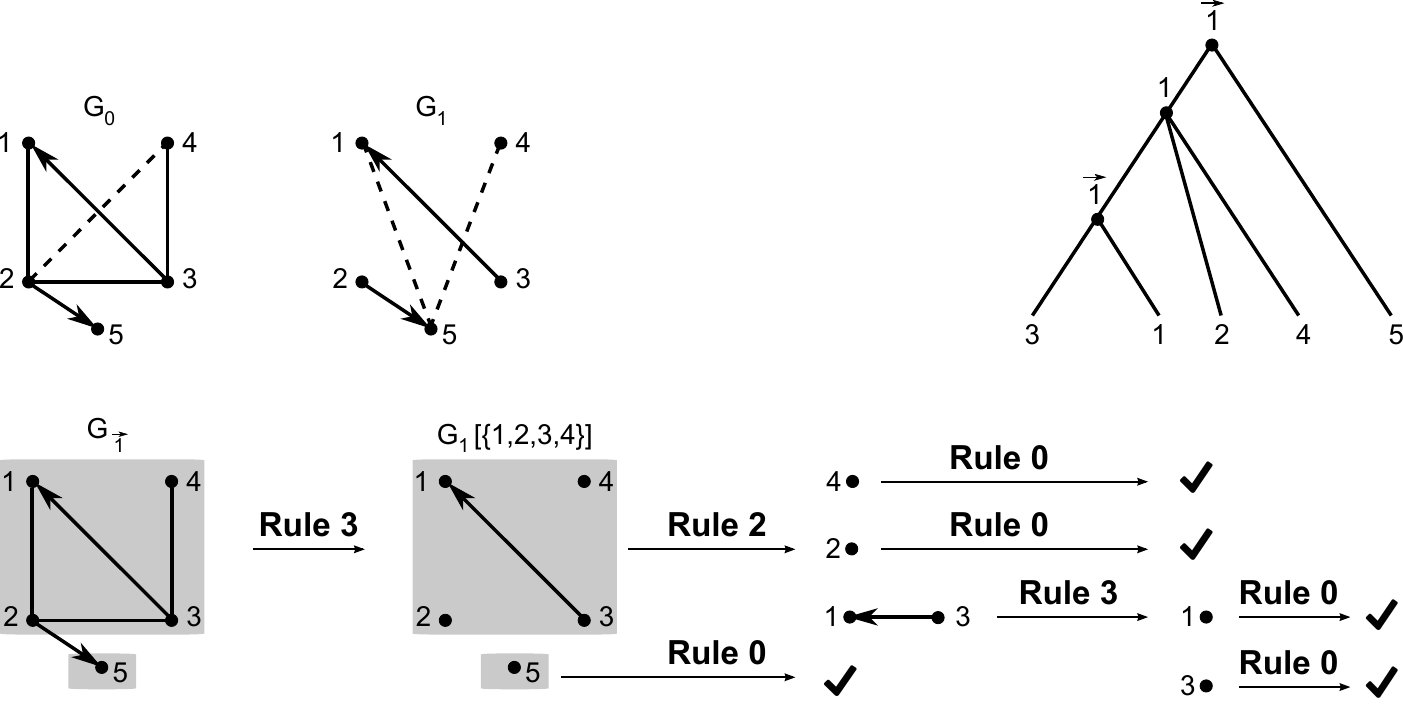}
\caption{Consider the partial relations $R_0 = \emptyset$,
         $R_1=\{[1,2],[2,3],[3,4]\}$, and $R_{\protect\Rone}=\{(3,1),(2,5)\}$,
         and the forbidden relations $F_0 = \{[2,4]\}$, $F_1 = \{[1,5],[4,5]\}$
         and $F_{\protect\Rone}=\emptyset$. Here $[x,y]\in R$ means that both
         $(x,y)$ and $(y,x)$ are contained in $R$. Left, the three graphs
         $G_0,G_1$ and $G_{\protect\Rone}$ as defined in Thm.\ \ref{thm:sat} are
         shown (arrows are omitted for symmetric arcs and dashed arcs correspond
         to forbidden pairs). Both $G_0$ and $G_1$ are connected. However,
         $G_{\protect\Rone}$ contains the two strongly connected components 
				 $C_1=\{1,2,3,4\}$ and $C_2=\{5\}$
         (highlighted with gray rectangles). Thus, Rule (3) can be applied. 
				  The graphs $G_0[C_1]$ and $G_{\protect\Rone}[C_1]$ are connected. 
					However, since $G_1[C_1]$ is disconnected one can apply Rule (2).
					 For the graph $G[C_2]$ only Rule (0) can be applied. 
				 The
         workflow in the lower part shows stepswise application of allowed rules
         on the respective components. The final cotree $(T,t)$ that explains
         the full set $\Hs= \{ \Rs_0 = \emptyset,
         R_1=\{[1,2],[2,3],[3,4],[1,4],[2,4]\},R_{\protect\Rone}=\{(3,1),(1,5),(2,5),(3,5),(4,5)\}$
         is shown in the upper right part.\vspace{-0.1in}
}
\label{fig:w-exmpl}
\end{figure}

In the absence of  forbidden relations,  Thm.\  \ref{thm:sat} immediately implies 
\begin{corollary}
The partial set $\mathcal{H} = \{R_0,R_1,R_{\Rone}\}$ is satisfiable if and only if
at least one of the following statements hold \vspace{-0.1in}
\begin{description}
\item[\ \textnormal{Rule (0):}] $|V| = 1$
	\item[\ \textnormal{Rule (1):}] (a) $G_0 \coloneqq (V,R_1\cup R_{\Rone})$ is disconnected and (b) each connected component of $G_0$ is satisfiable. 
	\item[\ \textnormal{Rule (2):}] (a) $G_1 \coloneqq (V,R_0\cup R_{\Rone})$ is disconnected and (b) for each connected component of $G_1$ is satisfiable.  
	\item[\ \textnormal{Rule (3):}] 
	(a) $G_{\Rone} \coloneqq (V,R_0\cup R_1\cup R_{\Rone})$ 
		  contains more than one strongly connected
			component, and (b) each strongly connected component of $G_{\Rone}$ is satisfiable. 
\end{description} \vspace{-0.07in}
\label{cor:sat}
\end{corollary} 
Thm.\  \ref{thm:character-satisf} and \ref{thm:sat} together with Cor.\ \ref{cor:sat} imply the following characterization of di-cographs. 
\begin{corollary}
$G=(V,E)$ is a di-cograph if and only if either \vspace{-0.07in}
\begin{description}
\item[\textnormal{(0)}] $|V| = 1$
\item[\textnormal{(1)}] $G$ is disconnected and each of its connected components are di-cograph.
\item[\textnormal{(2)}] $\overline{G}$ is disconnected and each of its connected components are di-cographs.
\item[\textnormal{(3)}] $G$ and $\overline{G}$ are connected, but  $G$ 
												contains more than one strongly connected component, each of them is a di-cograph. 
\end{description} \vspace{-0.1in}
\label{thm:cographNew}
\end{corollary} 
Thm.\  \ref{thm:sat} gives a characterization of satisfiable partial sets
$\mathcal{H} = \{ R_0,R_1, R_{\Rone}\}$ with respect to some forbidden sets
$\mathcal{F} =  \{ F_0, F_1,F_{\Rone} \}$. In the appendix, it is shown that
the order of applied rules has no influence on the correctness of the
algorithm. Clearly, Thm.\  \ref{thm:sat} immediately provides an algorithm
for the recognition of satisfiable sets $\mathcal{H}$, which is summarized in Alg.\ \ref{alg:build-cotree}.
Fig.\  \ref{fig:w-exmpl} shows an example of the application of Thm.\  \ref{thm:sat}
and Alg.\ \ref{alg:build-cotree}.
\begin{restatable}{theorem}{thmcorrectness}
	Let $\mathcal{H} = \{ R_0, R_1, R_{\Rone} \}$ be a partial set, and 
	$\mathcal{F} = \{ F_0, F_1, F_{\Rone} \}$ a forbidden set. Additionally, let
	$n = |V|$ and $m =|R_0 \cup R_1 \cup R_{\Rone} \cup F_0 \cup F_1 \cup F_{\Rone}|$.
	Then, Alg.\   \ref{alg:build-cotree} runs in $O(n^2 + nm)$ time and either: \vspace{-0.1in}
	\begin{enumerate}
		\item[(i)]	outputs a cotree $(T,t)$ that explains $\mathcal{H}$; or  
		\item[(ii)] outputs the statement ``$\mathcal{H}$ is not satisfiable w.r.t.\  $\mathcal{F}$''.
	\end{enumerate}
\label{thm:algo}
\end{restatable}
Alg.\  \ref{alg:build-cotree} provides a cotree; $(T;t)$, explaining a full
satisfiable set; $\Hs = \{ \Rs_0, \Rs_1, \Rs_{\Rone} \}$ extended from a given
partial set $\mathcal{H}$, 
such that $\Hs$ is satisfiable w.r.t.\ a forbidden set $\mathcal{F}$. 
Nevertheless, it can be shown that $\Hs$ can
easily be reconstructed from a given cotree $(T;t)$ in $O(|V|^2)$ time (see Appendix).

\renewcommand{\algorithmicrequire}{\textbf{Input:}}
\renewcommand{\algorithmicensure}{\textbf{Output:}}
\begin{algorithm}[tbp]
	\caption{Recognition of satisfiable partial sets $\mathcal{H}$ w.r.t. forbidden sets $\mathcal{F}$ and reconstruction of a cotree $(T,t)$ that explains $\mathcal{H}$.} 
	\label{alg:build-cotree} 
  \begin{algorithmic}[1] 
		\Require{Partial sets $\mathcal{H} = \{R_0, R_1, R_{\protect\Rone}\}$ and $\mathcal{F} = \{ F_0, F_1, F_{\protect\Rone} \}$ }
		\Ensure{ A cotree $(T;t)$ that explains $\mathcal{H}$, if one exists and $R_0 \cap F_0 = R_1 \cap F_1 = R_{\protect\Rone} \cap F_{\protect\Rone} = \emptyset$
	
						or 					the statement	
						 ``$\mathcal{H}$ is not satisfiable w.r.t.\  $\mathcal{F}$''}
		\If{$R_0 \cap F_0 = R_1 \cap F_1 = R_{\protect\Rone} \cap F_{\protect\Rone} = \emptyset$}\label{lin:valid-assumption}\label{lin:intersect-empty}
				 Call \texttt{BuildCotree}($V, \mathcal{H}, \mathcal{F}$) 
		\Else 
		\State Halt and output: ``$\mathcal{H}$ is not satisfiable w.r.t.\ $\mathcal{F}$''
		\EndIf
		
		\smallskip
	  \Function{\texttt{BuildCotree}}{$V, \mathcal{H}= \{R_0, R_1, R_{\protect\Rone} \}, \mathcal{F} = \{ F_0, F_1, F_{\protect\Rone} \}$ }
		
		\Comment{ $G_0$,  $G_1$ and $G_{\protect\Rone}$ are defined as in Thm. \ref{thm:sat} for given $\mathcal{H}$ and $\mathcal{F}$}
	  \If{$|V| = 1$} \Return{the cotree $((V, \emptyset), \emptyset)$}\label{lin:r0}
	  \ElsIf{$G_0$ (resp. $G_1$)  is	  disconnected}\label{lin:r1}
		  \State $\mathcal{C} \coloneqq $ the set of connected components $\{ C_1, \dots, C_k\}$ of $G_0$ (resp. $G_1$)
		  \State $\mathcal{T} \coloneqq \{\texttt{BuildCotree}(V[C_i],\ \mathcal{H}[C_i],\ 
						  \mathcal{F}[C_i]) \mid C_i \in \mathcal{C}\}$\label{lin:run-tree1}
		  \State \Return{the cotree from joining the cotrees in
									   $\mathcal{T}$ under a new root labeled $0$ (resp. $1$)}
	  \ElsIf{$G_{\protect\Rone}$ has more than one strongly connected component}\label{lin:r3}
		  \State $\mathcal{C} \coloneqq $ the set of strongly connected components $\{C_1, \dots, C_k\}$ of $G_{\protect\Rone}$
		  \State $\pi \coloneqq$ a topological order on the quotient $G / \{C_1, \dots, C_k\}$
		  \State $\mathcal{T} \coloneqq \{\texttt{BuildCotree}(V[C_i],\ \mathcal{H}[C_i],\ 
		  \mathcal{F}[C_i])\ |\ \text{ for all } C_i,\ i = 1, \dots, k\}$\label{lin:run-tree2}
		  \State \Return{the cotree $(T,t)$ obtained by joining the cotrees in $\mathcal{T}$	under a new root 

\hspace{1.2cm}									with label $\protect\Rone$, where $T_i$ is placed left from $T_j$ whenever $\pi(C_i) < \pi(C_j)$ 
 
}
		\Else
			\State Halt and output: ``$\mathcal{H}$ is not satisfiable w.r.t.\  $\mathcal{F}$''
	  \EndIf
	  \EndFunction
  \end{algorithmic}
\end{algorithm}

\vspace{-0.1in}
\section{Experiments}\vspace{-0.1in}

In this section, we investigate the accuracy of the recognition algorithm 
and compare recovered relations with known full sets that are obtained from simulated 
cotrees. The intended practical application that we have in mind, is to reconstruct
estimated homology relations. In this view, sampling random trees would not be sufficient, 
as the evolutionary history of genes and species tend to produce fairly balanced
trees. Therefore, we used the DendroPy \texttt{uniform\_pure\_birth\_tree} model \cite{sukumaran2010dendropy,hartmann2010sampling} 
to simulate 1000 binary gene trees for each of the three different leaf sizes $|L|\in \{25,50,100\}$. 
In addition, we randomly labeled the inner vertices of all trees as
``$0$'', ``$1$'' or ``$\Rone$'' with equal probability.

Each cotree $\mathcal{T} = (T;t)$ then represents a full
set $\Hs(\mathcal{T}) = \{ \Rs_0(\mathcal{T}), \Rs_1(\mathcal{T}), \Rs_{\Rone}(\mathcal{T}) \}$. 
For each of the full sets  $\Hs(\mathcal{T})$ and any two vertices
$x,y\in V$, the corresponding gene pairs $(x,y)$ and $(y,x)$ 
is removed from $\Rs_0(\mathcal{T})\cup \Rs_1(\mathcal{T}) \cup \Rs_{\Rone}(\mathcal{T})$ 
with a fixed probability $p\in \{0.1,0.2,\dots,1\}$.
Hence, for each $p\in 
\{0.1,0.2,\dots,1\}$ and each fixed leaf size $|L|\in \{25,50,100\}$, 
we obtain 1000 partial sets $\mathcal{H} = \{R_0,R_1,R_{\Rone} \}$.
We then use Alg.\  \ref{alg:build-cotree} on each partial set $\mathcal{H}$, to
obtain a cotree $\widetilde{\mathcal{T}} = (\widetilde T; \tilde t)$ explaining 
the full set $\Hs(\widetilde{\mathcal{T}}) = \{\Rs_0(\widetilde{\mathcal{T}}), R_1(\widetilde{\mathcal{T}}),
R_{\Rone}(\widetilde{\mathcal{T}})\}$.

\begin{figure}[t]
\centering     
\subfigure{\label{fig:1st}\includegraphics[viewport = 85 363 400 700, clip, width=0.35\textwidth]{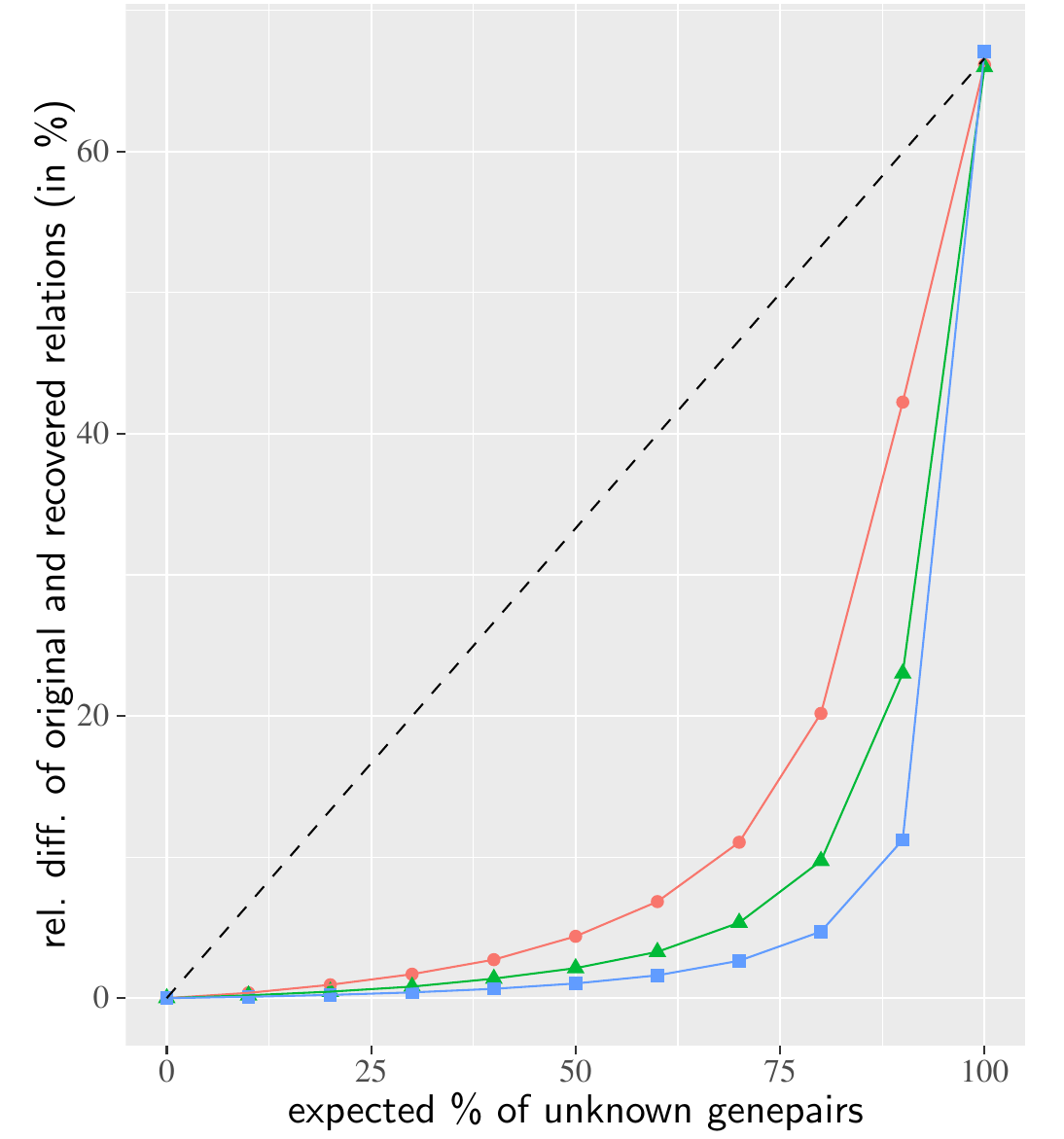}	}
\subfigure{\label{fig:forbid}\includegraphics[viewport = 85 363 400 700, clip, width=0.35\textwidth]{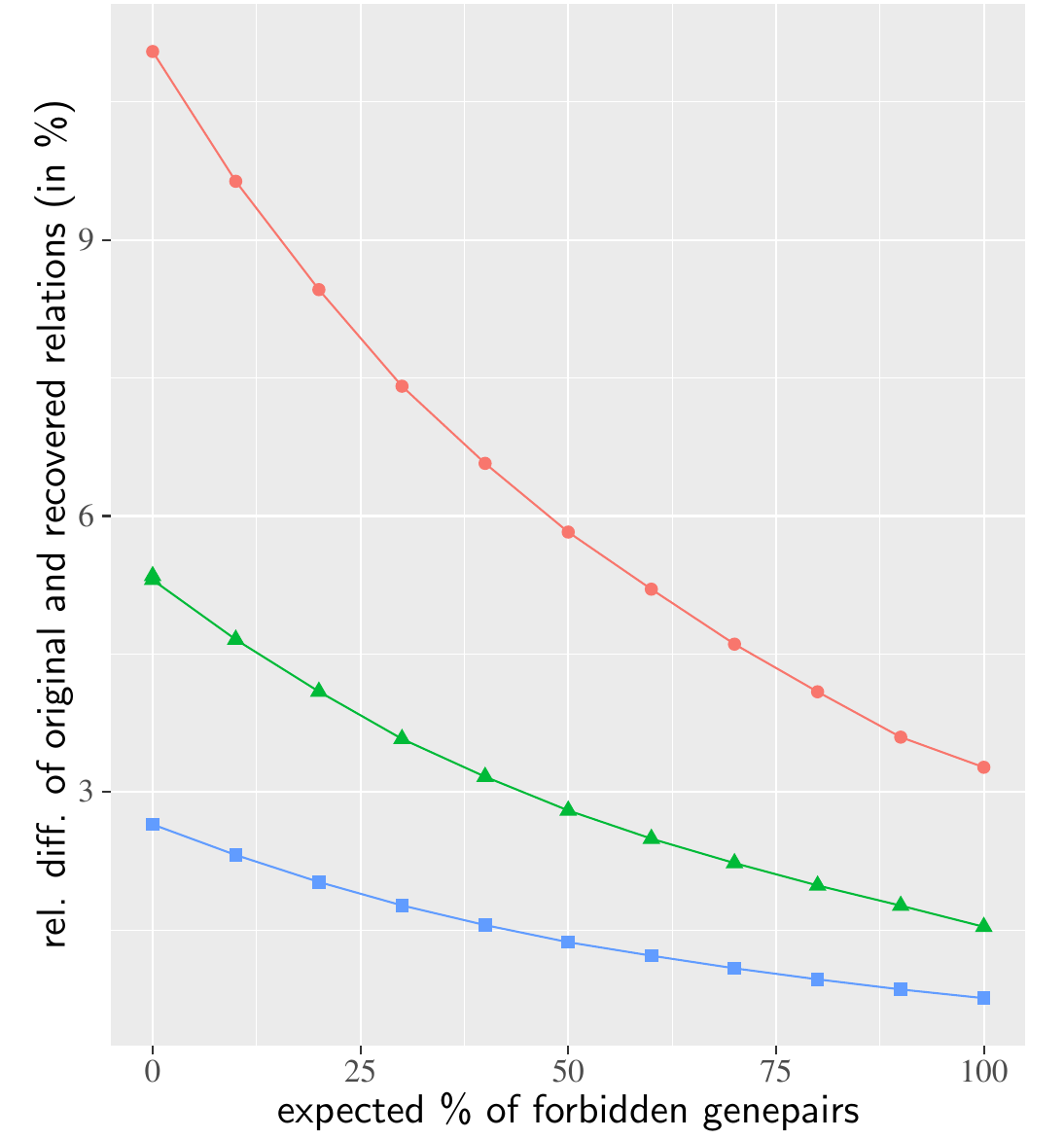}}
\subfigure{\label{fig:forbid2}	\includegraphics[viewport = 400 380 435 700, clip, width=0.05\textwidth]{sizes_reldiff_viewport.pdf}		}
\caption{Shown are the average relative differences of original and recovered relations
				 depending on the size of unknowns (left) and the size of additional forbidden relations (right). \vspace{-0.1in} }
\label{fig:plot1}
\end{figure}

Fig.\  \ref{fig:plot1}(left) shows the average relative difference of the original full
set $\Hs(\mathcal{T})$ and the recovered full sets
$\Hs(\widetilde{\mathcal{T}})$ for each instance. The dashed line in the plots
of Fig.\  \ref{fig:plot1}(left) shows the expected relative difference when each unknown
gene pair is assigned randomly to one of the relations in the partial set
$\mathcal{H}$. As expected, the relative differences increases with the number
of unassigned leaf pairs. Somewhat surprisingly, even if 80\% of pairs of leaves are
expected to be unassigned in $\mathcal{H}$, it is possible to averagely recover
79.8\% - 95.3\% of the original relations. Moreover, the plot in Fig.\ 
\ref{fig:plot1}(left) also suggest that the accuracy of recovered homology
relations increases with the input size, i.e., the number of leaves. To
explain this fact, observe that the number of constraints given by the full set of
homology relations on some leaf set $L$ is $O(|L|^2)$. Conversely, the number of
inner vertices in a tree only increases linearly with $L$, $O(|L|)$. Hence, on
average the number of constraints given on the labeling of an internal vertex in
the gene tree is $O(|L|^2)/O(|L|)= O(|L|)$. Note, Alg.\ 
\ref{alg:build-cotree} constructs cotrees and hence, if there are more leaves,
then there are also more constraints for the rules (labeling of the inner
vertices) that are allowed to be applied. Therefore, with an increasing number
of leaves the correct relation between unassigned pairs of leaves in
$\mathcal{H}$ are already determined.

Fig.\  \ref{fig:plot1}(right) shows the impact of additional forbidden relations
for the instances where we have removed pairs $(x,y)$ with probability $p=0.7$.
For each of the partial sets $\mathcal{H}$, 
we have chosen two vertices $x$ and $y$ where neither  $(x,y)$ nor $(y,x)$
is contained in $\mathcal{H}$ with probability $p'\in \{0.1, 0.2, \dots, 1\}$
and assigned $(x,y),(y,x)$ to a forbidden relation $F_i$ such that 
$(x,y),(y,x)\in \Rs_j(\mathcal{T})$ with $i,j\in\{0,1,\Rone\}$ implies that $i\neq j$.
In other words, if $(x,y),(y,x)$ are assigned to $F_i$ with $i\in\{0,1,\Rone\}$ 
then $(x,y),(y,x)$ were not in the original set $\Rs_j(\mathcal{T})$. 
The latter is justified to ensure satisfiability of the partial relations w.r.t.
the forbidden relations. Again, we compared the relative difference of the original full
sets $\Hs(\mathcal{T})$ and the recovered full sets
$\Hs(\widetilde{\mathcal{T}})$ computed with Alg.\  \ref{alg:build-cotree}.
The plot shows that, with an increasing number of forbidden leaf pairs, the relative difference 
decreases. Clearly, the more leaf pairs are forbidden the more of such leaf pairs
are not allowed to be assigned to one of the relations. Therefore, the degree
of freedom for assigning a relation to an unassigned pair decreases with an
increase of the number of forbidden pairs.

One factor that may affect the results of the plots shown in Fig.\ 
\ref{fig:plot1} is the order in which rules are chosen when more than
one rule is satisfied. By construction, Alg.\  \ref{alg:build-cotree} 
fixes the order of applied rules as follows: first Rule (1), then Rule (2), then Rule (3). 
In other words, when possible 
the trees for the satisfiable (strongly) connected components 
are first joined by a common root labeled ``$1$''; if this does not
apply, then with common root labeled ``$0$'', and ``$\Rone$'' otherwise. 
To investigate this issue in more details, we modified Alg.\  \ref{alg:build-cotree} 
so that either a different fixed rule order or a random rule order 
 is applied.

\begin{figure}[tbp]
\centering     
\subfigure{\label{fig:a}\includegraphics[viewport = 85 363 367 700, clip, width=0.35\textwidth]{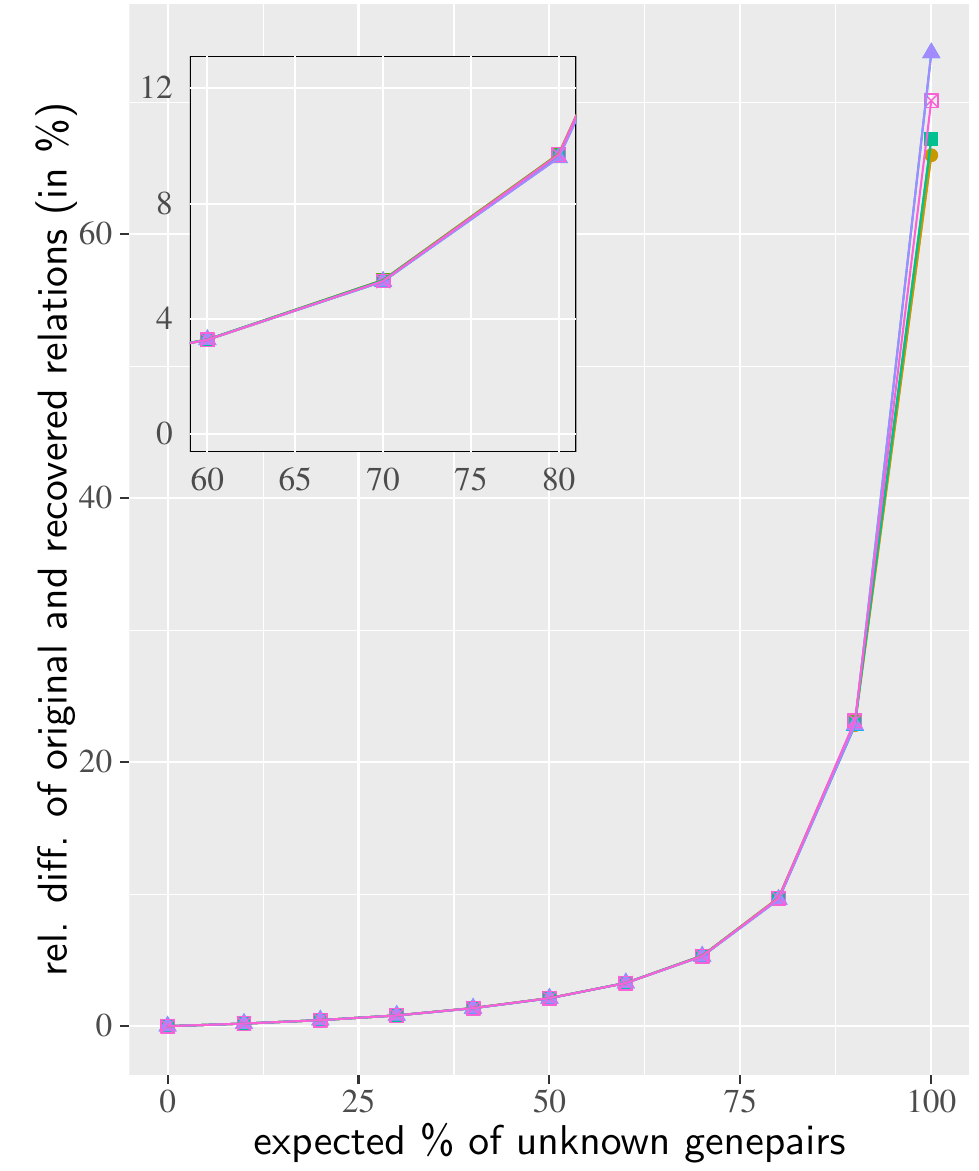}	}
\subfigure{\label{fig:b}\includegraphics[viewport = 85 363 367 700, clip, width=0.35\textwidth]{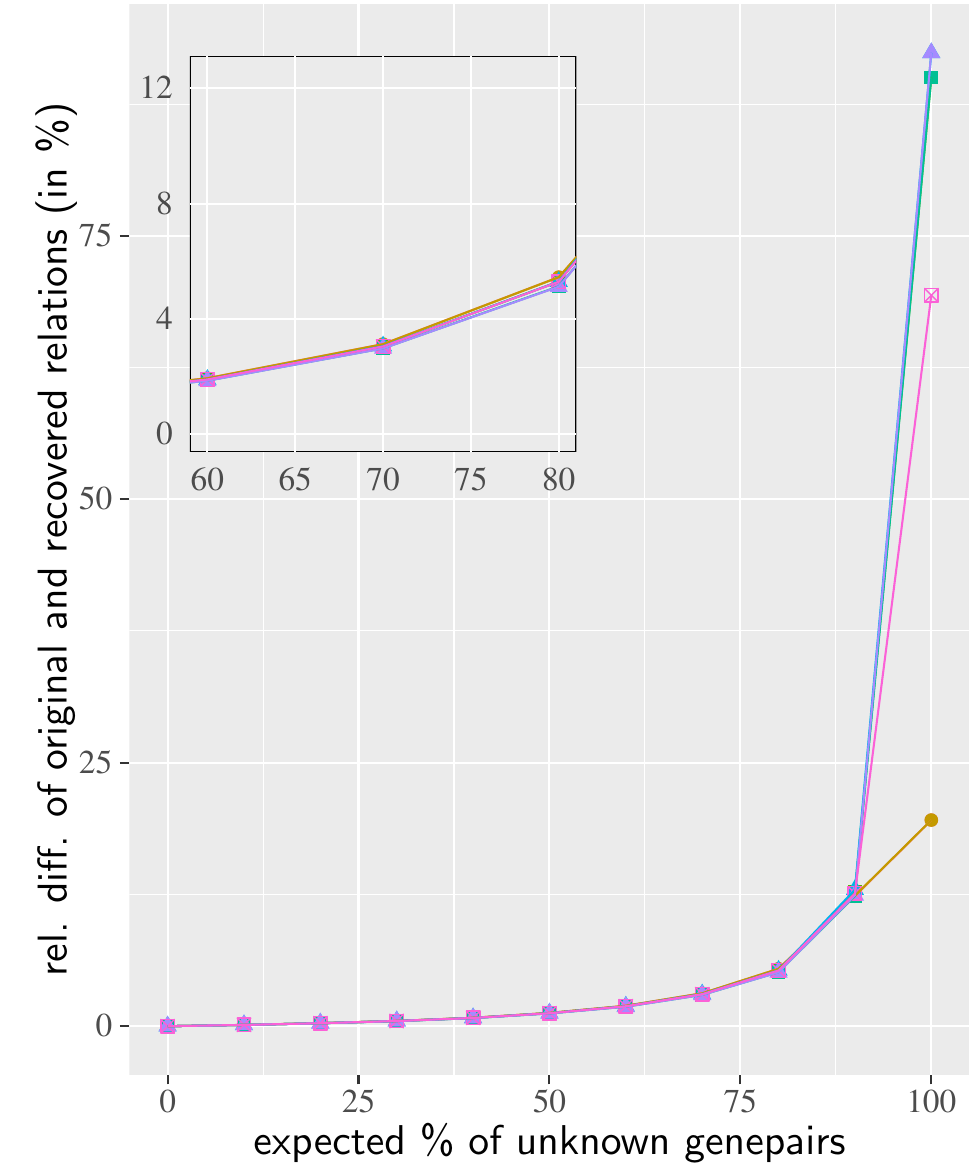}}
\subfigure{\label{fig:d}	\includegraphics[viewport = 367 363 436 700, clip, width=0.1\textwidth]{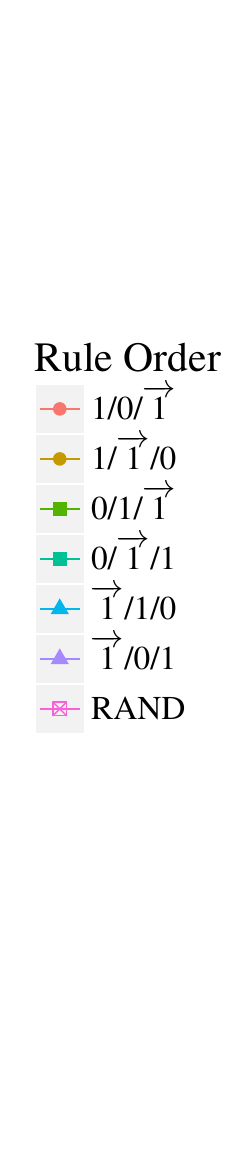}		}
	\caption{Shown are the average relative differences of original and recovered relations, 
					 depending on the percentage of unassigned pairs and the rule order 
					 for cotrees with uniform (left) and skewed (right) label distribution. \vspace{-0.1in}}
\label{fig:plot2}
\end{figure} 

Fig.\  \ref{fig:plot2}(left) shows the plot for the partial relations for the fixed
leaf size $|L| = 50$. The rule orders are shown
in the legend of Fig.\  \ref{fig:plot2}. Here, $X/Y/Z$, with $X,Y,Z\in \{1\ \widehat
= \textrm{ Rule(2)}, 0\ \widehat = \textrm{ Rule(1)},\Rone\ \widehat = \textrm{	Rule(3)}\}$ being distinct;
indicates that first rule $X$ is checked, then  rule $Y$ and, if both are not applicable, then rule $Z$ is used. 
\emph{RAND} means that each of the allowed rules are chosen with equal probability. 
As one can observe, the rule order does not have a significant impact 
on the quality of the recovered full sets. 
This observation might be explained by the fact that we have
used a random assignment of events $0,1$ and $\Rone$ for the vertices of
the initial simulated trees $T$. 

To investigate this issue in more detail, we additionally 
used 1000 unlabeled simulated trees $T$ with $|L| = 50$ 
and assigned to each vertex $v$ a label $t(v)=1$ with a probability $p=0.8$, 
label $t(v)=\Rone$ with $p=0.1$ and label $t(v)=0$ with $p=0.1$
Again, each resulting cotree $\mathcal{T} = (T;t)$ represents a full
set $\Hs(\mathcal{T}) = \{ \Rs_0(\mathcal{T}), \Rs_1(\mathcal{T}), \Rs_{\Rone}(\mathcal{T}) \}$
from which we obtain partial sets $\mathcal{H} = \{R_0,R_1,R_{\Rone} \}$
as for the other instances. 
Fig.\  \ref{fig:plot2}(right) shows the resulting plots. 
As it can be observed, even a quite skewed distribution of labels
in the cotrees and the choice of rule order does not have an effect on 
the quality of the recovered full sets. 

\emph{In summary, the results show that it suffices to have only a very few but correct pairs
			of the original relations to reconstruct most of them.}

\section*{Acknowledgment}\vspace{-0.1in}
This contribution is supported in part by the Independent Research Fund
				Denmark, Natural Sciences, grant DFF-7014-00041.
\vspace{-0.15in}

\clearpage
\section*{\appendixname}

\subsection*{Proof of Lemma \ref{lem:satisf-partition}}

\begin{proof}
		Assume that $\mathcal{H}$ is satisfiable w.r.t.\ $\mathcal{F}$. 
		Hence,  $\mathcal{H}$ can be extended to  
		a full set $\Hs = \{ \Rs_0, \Rs_1, \Rs_{\Rone} \}$ that is
		satisfiable w.r.t.\ $\mathcal{F}$. Thus, there is a cotree $\mathcal{T}=({T},{t})$
		that explains $\mathcal{\Hs}$ and, in particular,  $\Gs = (V,\Rs_1\cup \Rs_{\Rone})$ must be
	  a di-cograph.

		Let $\{C_1, \dots, C_k\}$, $k\geq 1$ be a partition of $V$. 
		Lemma \ref{lem:binary-induced} implies that $\Gs[C_i]$ is a di-cograph for any $i\in \{1,\dots,k\}$. 
		Since $\Rs_0\cap F_0 = \Rs_1\cap F_1 =  \Rs_{\Rone}\cap F_{\Rone} = \emptyset$,
		 the di-cograph $\Gs$ as well as the
		relation $\Rs_0$ does not contain arcs in $F_1, F_0$ and	$F_{\Rone}$. Therefore, 
		$\Gs[C_i]$ as well as the
			relation $\Rs_0[C_i]$ does not contain arcs in $F_1, F_0$ and	$F_{\Rone}$.
		This implies that  the full set 
		$\Hs[C_i]$ restricted to the vertices in $C_i$ is satisfiable w.r.t.\
		$\mathcal{F}[C_i]$. Hence, there is an extension $\Hs[C_i]$ 
		of the partial set $\mathcal{H}[C_i]$  that is satisfiable w.r.t. $\mathcal{F}[C_i]$
		and therefore,  $\mathcal{H}[C_i]$ is satisfiable w.r.t. $\mathcal{F}[C_i]$.

		Now assume that for any partition $\{C_1, \dots,	C_k\}$ of $V$ the 
		induced sub-relations in
		$\mathcal{H}[C_i]$ are satisfiable w.r.t.\ $\mathcal{F}[C_i]$.
		Thus, for the trivial partition $\{V\}$ of $V$, we have
		$\mathcal{H}[V] = \mathcal{H}$ is satisfiable  w.r.t.\ $\mathcal{F}[V] = \mathcal{F}$, 
		which completes		the proof. 
\qed
\end{proof}

\subsection*{Proof of Theorem \ref{thm:sat}}
\begin{proof}
"$\Leftarrow$:" Assume that at least one of the conditions of Rule (0-3) is satisfied. 
We will show that one can extend $\mathcal{H}$ to a full set $\Hs$ that is satisfiable
w.r.t.\ $\mathcal{F}$.

Clearly, if Rule (0) is satisfied, then $\mathcal{H}$ is already full and the
cotree that corresponds to the single-vertex graph $K_1$
 explains $\mathcal{H}$. Hence,
$\mathcal{H}$ is satisfiable w.r.t.\ $\mathcal{F}$. In what follows, we will
therefore assume that $|V| \geq 2$.

Assume that Rule (1) is satisfied, and let $C_1, \dots, C_k$, $k \geq 2$,
be the connected components of $G_0$. Since $\mathcal{H}[C_i]$ is
satisfiable w.r.t.\ $\mathcal{F}[C_i]$ for all $i \in \{1,\dots,k\}$, we can
extend $\mathcal{H}[C_i]$ to a full set $\Hs[C_i]$ that is satisfiable
w.r.t.\ $\mathcal{F}[C_i]$. Hence, there exists a cotree $(T_i, t_i)$
representing $\Hs[C_i]$.
Let $\mathcal{H}' = \{ R'_0,\Rs_1, \Rs_{\Rone}\}$ be the resulting
partial set obtained from extending $\mathcal{H}$, such that
$\Hs[C_i]$ is full and satisfiable w.r.t.\ $\mathcal{F}[C_i]$ for
all $i \in \{ 1, \dots, k \}$, and let $(T_1, t_1), \dots, (T_k, t_k)$ be
the cotrees explaining $\Hs[C_1], \dots, \Hs[C_k]$, respectively. 

For any remaining unassigned pair $(x,y) \not\in R'_0\cup \Rs_1\cup
R_{\Rone}^{\star,sym}$, the vertices $x$ and $y$ must be contained in two
distinct connected components $C_i$ and $C_j$ of $G_0$. Now, we extend
$\mathcal{H'}$ to a full set $\Hs$ by adding all unassigned pairs to $R'_0$
to obtain the set $\Rs_0$. We continue to show that $\Hs= \{ \Rs_0,\Rs_1,
\Rs_{\Rone}\}$ is satisfiable w.r.t.\ $\mathcal{F}$. Let $x$ and $y$ be two
vertices that reside in distinct connected components $C_i$ and $C_j$ of $G_0$.
Clearly, $(x,y)\notin F_0\cup R_1 \cup R_{\Rone}^{sym}$, as otherwise, $x$ and
$y$ would be in the same connected component of $G_0$. By construction,
$(x,y)\notin F_0\cup \Rs_1 \cup R_{\Rone}^{\star,sym}$ and $F_0\cap \Rs_0= \emptyset$.

Now, we join cotrees $(T_1,t_1), \dots, (T_k,t_k)$ the under a new root $\rho_T$
with label $0$ to obtain the cotree $(T,t)$. 
Note, any subtree rooted at a child of $\rho_T$ explains
$\Hs[C_i]$ for some connected component $C_i$ of $G_0$. Thus,
$(x,y) \in \Rs_0$ if and only if $t(\lca(x,y)) = 0$.
It follows that
$(T,t)$ explains the full set $\Hs = \{\Rs_0,\Rs_1,\Rs_{\Rone}\}$ 
and $\Rs_0 \cap F_0 = \Rs_1 \cap F_1 = \Rs_{\Rone} \cap F_{\Rone} = \emptyset$. Therefore
$\Hs$ is satisfiable w.r.t.\ $\mathcal{F}$. Hence, $\mathcal{H}$ is
satisfiable w.r.t.\ $\mathcal{F}$. 

The same arguments can be applied on the connected components of $G_1$ if Rule
(2) is satisfied by extending every unassigned pair $x,y\in V$ between two
distinct connected components of $G_1$ to $1$ and  setting
$t(\rho_T) \coloneqq 1$. Again, $\mathcal{H}$ is satisfiable w.r.t.\
$\mathcal{F}$.

Finally, assume that Rule (3) is satisfied, and let $C_1, \dots, C_k$, $k \geq
2$, be the strongly connected components of $G_{\Rone}$. As before, since
$\mathcal{H}[C_i]$ is satisfiable w.r.t.\ $\mathcal{F}[C_i]$ for all $i \in \{ 1,
\dots, k \}$, the induced subgraph $\mathcal{H}[C_i]$ can be extended to a full set that is satisfiable
w.r.t.\ $\mathcal{F}[C_i]$. Again, we let $\mathcal{H}' = \{ \Rs_0, \Rs_1,
R_{\Rone}'\}$ be the resulting partial set obtained from
$\mathcal{H}$, such that $\mathcal{H}'[C_i]$ is a full satisfiable set w.r.t.\
$\mathcal{F}[C_i]$, for all $i \in \{ 1, \dots, k \}$. Let $(T_1,t_1),
\dots, (T_k,t_k)$ be the corresponding cotrees explaining $\mathcal{H}'[C_1],
\dots, \mathcal{H}'[C_k]$, respectively.
In what follows, we want to show that
the cotree $(T,t)$ with root $\rho_T$ with $\Rone$ that is obtained by
joining the cotrees $(T_1,t_1), \dots (T_k, t_k)$ under the root $\rho_T$
in a particular order explains a full set $\Hs = \{\Rs_0,\Rs_1,\Rs_{\Rone}\}$
that extends $\mathcal{H'}$ and is satisfiable w.r.t.\  $\mathcal{F}$. 
To this end, however, we need to investigate the structure of the relations 
$R_0,R_1,R_{\Rone}$ and $\Rs_0, \Rs_1,R_{\Rone}'$ in some more detail. 

Let $Q \coloneqq G_{\Rone}/ \{C_1, \dots, C_k \}$ be the quotient graph. 
By definition, $Q$ is a DAG, and thus, there exists a topological
order $\pi$ on $Q$ such that for any arc $(C_i,C_j) \in E(Q)$, it holds that $\pi(C_i) <
\pi(C_j)$. W.l.o.g.\ assume that $\pi$ is already given by the ordering of $C_1,
\dots, C_k$, i.e., $\pi(C_i) \coloneqq i$ for all $i
\in \{ 1, \dots, k \}$.

By construction, for any unassigned pair $(x,y) \not\in \Rs_0 \cup \Rs_1 \cup
R'_{\Rone}$ of $\mathcal{H}'$, the two vertices $x$ and $y$ are contained in two
distinct strongly connected components $C_i$ and $C_j$ of $G_{\Rone}$. Now, we extend
$\mathcal{H}'$ to a full set $\Hs=\{\Rs_0, \Rs_1, \Rs_{\Rone}\}$,
by adding all unassigned pairs $(x,y)$ with $x\in C_i$ and $y\in C_j$ to the
relation $R'_{\Rone}$ if and only if $i < j$, to obtain $\Rs_{\Rone}$. 

Let $C_i$ and $C_j$ be distinct strongly connected components of $G_{\Rone}$.
We show now that $(x,y)\in \Rs_{\Rone}$ for all $x\in C_i$ and $y\in C_j$
if and only if $i < j$.
By construction, the latter is satisfied for all $(x,y)\in \Rs_{\Rone}\setminus R'_{\Rone}$.
Let $(u,v)$ be an already assigned pair that is contained in $\Rs_0 \cup \Rs_1 \cup
R'_{\Rone}$ such that $u\in C_i$ and $v\in C_j$. Therefore, 
$(x,y)\in R_0\cup R_1\cup R_{\Rone}$ since no pairs $(x,y)$ have been added 
between vertices $x\in C_i$ and $y\in C_j$ to obtain $\Rs_0 \cup \Rs_1 \cup
R'_{\Rone}$. But then, $(x,y)\in R_{\Rone}$, as otherwise symmetry of $R_0$ and $R_1$ would imply
that $u$ and $v$ are in the same strongly connected component of $G_{\Rone}$. 
Furthermore, we added only pairs $(a,b)$ to $R_{\Rone}$ to obtain  $R'_{\Rone}$
where $a$ and $b$ are contained within the same strongly connected component of $G_{\Rone}$.
Thus, for all  $(u,v) \in \Rs_0 \cup \Rs_1 \cup R'_{\Rone}$ with
$u$ and $v$ in distinct strongly connected components $C_1,\dots, C_k$
it holds that $(u,v) \in R_{\Rone}$.
In summary,  we added the unassigned pair $(x,y)$ with  $x\in C_i$ and $y\in C_j$ 
to $R'_{\Rone}$ if and only if $\pi(C_i) = i < j = \pi(C_j)$. In addition, 
any already assigned pair  $(u,v)$ with  $u\in C_i$ and $v\in C_j$ 
must be contained in $R_{\Rone}$, which 
implies that $(C_i, C_j) \in E(Q)$, and hence, $\pi(C_i) = i < j = \pi(C_j)$. 
Therefore, $(x,y)\in \Rs_{\Rone}$ if and only if $i < j$.
Moreover, the latter arguments also ensure that the sets $\Rs_0, \Rs_1$ and $\Rs_{\Rone}$ 
are pairwisely disjoint.

We continue to show that $\Rs_{\Rone} \cap F_{\Rone} = \emptyset$. 
By construction, $R'_{\Rone}[C_i]$ is satisfiable
w.r.t.\  $F_{\Rone}[C_i]$ for all $i\in \{1,\dots,k\}$. 
Hence,  $R'_{\Rone} \cap F_{\Rone} = \emptyset$. 
Thus, it remains to show that for all $(x,y) \in \Rs_{\Rone}\setminus R'_{\Rone}$
it holds that $(x,y)\not\in F_{\Rone}$. 
Let $(x,y) \in \Rs_{\Rone}\setminus R'_{\Rone}$.
Thus, $x\in C_i$ and $y\in C_j$ for some strongly connected components of 
$G_{\Rone}$ with $i<j$. 
Since $(x,y) \notin R'_{\Rone}$, we also have $(x,y) \not\in R_{\Rone}$. 
Assume for contraction that $(x,y) \in F_{\Rone}$. Hence, 
$(y,x) \in \overleftarrow{F_{\Rone}}$.
However, this would imply that $(C_j, C_i) \in E(Q)$,
meaning that $i > j$ and therefore $(x,y) \not\in \Rs_{\Rone}$; a contradiction.

We are now in the position to create a cotree $(T,t)$ that explains $\Hs = \{\Rs_0,\Rs_1,\Rs_{\Rone}\}$. 
To this end, we join the cotrees $(T_1,t_1), \dots (T_k, t_k)$ under common root $\rho_T$
with label $\Rone$. Moreover,  cotrees $(T_1,t_1), \dots (T_k, t_k)$ are added to $(T,t)$ such that 
any leaf of $T_i$ is left of any leaf of $T_j$ iff $i < j$. 
By construction $(T,t)$ explains all sub-relations in $\Hs[C_i]$, $1\leq i\leq k$.
For any other pair $(x,y)$ where $x$ and $y$ are contained in different
strongly connected components $C_i$ and $C_j$ of $G_{\Rone}$, respectively,
we have $(x,y)\in \Rs_{\Rone}$, $(y,x)\not\in \Rs_{\Rone}$, $(x,y) \not\in
F_{\Rone}$ and $i<j$. As shown above, $\Rs_0\cap F_0 = \Rs_1 \cap F_1 = \Rs_{\Rone} \cap F_{\Rone} = \emptyset$.
By construction,  $t(\lca_T(x,y))=t(\rho_T) = \Rone$ and $x$ is placed left of $y$ in $T$.
Therefore, any pair $(x,y)\in \Rs_{\Rone}$ is explained by $(T,t)$. 
In summary, $(T,t)$ explains $\Hs$ and $\Hs$ 
is satisfiable w.r.t.\ $\mathcal{F}$. Hence,  $\mathcal{H}$
is satisfiable w.r.t.\ $\mathcal{F}$.

``$\Rightarrow:$'' 
	Assume that $\mathcal{H} = \{R_0,R_1,R_{\Rone}\}$ is satisfiable w.r.t.\ $\mathcal{F}$. 
  Trivially, Rule (0) is satisfied in the case that $|V| = 1$. Hence, let $|V|
	\geq 2$. By assumption, $R_0,R_1$ and $R_{\Rone}$ can be extended to obtain a full
	satisfiable set $\Hs = \{\Rs_0,\Rs_1,\Rs_{\Rone}\}$ such that $\Rs_0 \cap F_0 = 
	\Rs_1\cap F_1 = \Rs_{\Rone} \cap F_{\Rone} = \emptyset$. 
	Hence, there is a cotree  $(T;t)$ that explains $\Hs$. Consider the root
	$\rho_T$ of $T$ with children $v_1,\dots, v_k$ and the particular leaf sets 
  $L(v_1), \dots, L(v_k)$. 
  Note that, by definition of cotrees each inner vertex of $(T;t)$, and in
  particular the root $\rho_T$ has at least two children.

	Assume that $\rho_T$ is labeled $0$. Since $(T;t)$ explains $\Hs$, 
	for any $x\in L(v_i)$ and $y\in L(v_j)$, $i\neq j$ we have 
	$(x,y),(y,x)\in \Rs_0$. 	Since additionally $\Rs_0 \cap
	F_0=\emptyset$, the digraph $\Gs_0 \coloneqq (V,\Rs_1\cup \Rs_{\Rone} \cup F_0)$
	must be disconnected. 
	Since $G_0$ is a subgraph of $\Gs_0$, the graph 
	 $G_0$ is disconnected as well.  By Lemma
	\ref{lem:satisf-partition},  $\mathcal{H}[C]$ 
	is satisfiable w.r.t.\ $\mathcal{F}[C]$ for 
	any connected components $C$ of $G_0$.
	Hence, Rule (1) is satisfied.

	If $\rho_T$ is labeled $1$, then we can apply analogous arguments
	to see that $\Gs_1$ and $G_1$ are disconnected and that  Rule (2) is satisfied. 

	Now assume that $\rho_T$ is labeled $\Rone$.  Let  $x\in L(v_i)$
  and $y\in L(v_j)$, $i\neq j$ and assume that $x$ is placed left of $y$ in $T$.
	Since $(T,t)$ explains $\Hs$, we have $(x,y) \in \Rs_{\Rone}$ and   $(y,x) \notin \Rs_{\Rone}$.
	Moreover,  $(x,y) \in \Rs_{\Rone}$ implies that $(x,y) \notin F_{\Rone}$ and thus, 
	$(y,x)\notin\overleftarrow{F_{\Rone}}$. The latter together with the disjointedness
	of the sets $\Rs_0,\Rs_1,\Rs_{\Rone}$ implies that 
	$(x,y) \in E(\Gs_{\Rone})$ and   $(y,x) \notin E(\Gs_{\Rone})$, where
	 $\Gs_{\Rone} \coloneqq (V,\Rs_0\cup \Rs_1\cup \Rs_{\Rone} \cup \overleftarrow{F_{\Rone}})$.
  Since the latter is satisfied for all elements in any of the leaf sets 
  $L(v_1), \dots, L(v_k)$, we immediately obtain that 
	$\Gs_{\Rone} = \Gs_{\Rone}[L(v_1)] \oslash \dots \oslash
	\Gs_{\Rone}[L(v_k)]$. Hence,  $\Gs_{\Rone}$ contains more than one
	strongly connected component. Since $G_{\Rone}$ is a subgraph of $\Gs_{\Rone}$, it
	follows that $G_{\Rone}$ contains more than one strongly connected
	component.  By Lemma \ref{lem:satisf-partition}, 
	$\mathcal{H}[C]$ is satisfiable w.r.t.\ $\mathcal{F}[C]$ for 
	all strongly connected components $C$ of $G_{\Rone}$.
	Hence, Rule (3) is satisfied.
\qed
\end{proof}


\subsection*{Correctness of Recognition Algorithm \ref{alg:build-cotree}}

In this part we provide a polynomial-time algorithm for the 
recognition of partial satisfiable sets $\mathcal{H}$ and the 
reconstruction of respective extended sets $\Hs$ and a cotree $(T,t)$
that explains $\mathcal{H}$, in case $\mathcal{H}$ is satisfiable. 
Thm.\  \ref{thm:sat} gives a characterization of satisfiable partial sets
with respect to some forbidden set. 
However, in order to design a
polynomial-time algorithm, we have to decide which rule can be applied at which
step. Note, several rules might be fulfilled at the same time, that is,
the graph $G_0$ and $G_1$ can be disconnected 
and $G_{\Rone}$ can contain more than one
strongly connected component at the same time.
Since sub-condition (b) for each Rule (1-3) is recursively defined, 
we might end in a non-polynomial algorithm, if 
a ``correct'' choice of the rule would be important.
A trivial example where all rules can be applied at the same time
is given by the partial set of empty relations $\mathcal{H} = \{ \emptyset,
\emptyset, \emptyset \}$. 
Interestingly, Lemma \ref{lem:satisf-partition} immediately implies that
it does not matter which applicable rule is chosen to obtain
$\Hs$ that is  satisfiable w.r.t.\ $\mathcal{F}$.
\begin{corollary}\label{cor:any-rule}
	Let $\mathcal{H} = \{ R_0, R_1, R_{\Rone} \}$ be a partial set, and $\mathcal{F} = \{ F_0, F_1, F_{\Rone} \}$ a forbidden set. Assume
	that Rule (i).a and (j).a with $i,j\in \{1,2,3\}$ (as given in Thm.
	\ref{thm:sat})
 are satisfied. 

	Then, $\mathcal{H}$ is satisfiable w.r.t.\ $\mathcal{F}$ if and only if Rule i.b and
	j.b are satisfied. 	
\end{corollary}
Thm.\  \ref{thm:sat} together with Corollary \ref{cor:any-rule} immediately yields
a polynomial time recursive algorithm to determine satisfiability of a homology set with
respect to a forbidden set. Moreover, if $\mathcal{H}$ is satisfiable w.r.t\
$\mathcal{F}$ the proof of Thm.\  \ref{thm:sat} describes
how to construct a cotree explaining a full satisfiable set w.r.t\
$\mathcal{F}$ extended from $\mathcal{H}$. 
The algorithm is summarized in 
Alg.\  \ref{alg:build-cotree}.

\thmcorrectness*
\begin{proof}
	The correctness of Alg.\  \ref{alg:build-cotree} follows from Thm.\  \ref{thm:sat}
	and Corollary \ref{cor:any-rule}. To be more precise, Alg.\ 
	\ref{alg:build-cotree} checks whether condition Rule (0) or Rule (1-3)(a),
	as defined in Thm.\  \ref{thm:sat}, is satisfied.
	
	As a preprocessing step, the algorithm first checks if $R_0 \cap F_0 = R_1
	\cap F_1 = R_{\Rone} \cap F_{\Rone} = \emptyset$ (Line
	\ref{lin:valid-assumption}). If this is not the case, clearly
	$\mathcal{H}$ cannot be satisfiable w.r.t.\ $\mathcal{F}$, and hence the
	algorithm stops. If	no forbidden pairs are present in $\mathcal{H}$, 
	we will try and build the cotree explaining a full set,
	extended from $\mathcal{H}$, that is satisfiable w.r.t.\ $\mathcal{F}$. Note, Corollary
	\ref{cor:any-rule} implies that the order of applied rules does not matter.

	If Rule (0) is satisfied, i.e., $|V| = 1$, then $\mathcal{H}$ is already full and
	satisfiable w.r.t. $\mathcal{F}$,
	and $(T;t) = ((V, \emptyset);\emptyset)$
	is a valid cotree explaining $\mathcal{H}$, and thus returned on Line
	\ref{lin:r0}.

	If Rule (1a) (resp.\ (2a)) is satisfied (Line \ref{lin:r1}), then
	Alg.\  \ref{alg:build-cotree} is called recursively on each of the
	connected components defined by $G_0$ ($G_1$ resp.) to verify that Rule
	(1b) (resp.\ (2b)) is	satisfied. 
	If the connected components are indeed satisfiable, the obtained cotrees
	are joined into a single cotree explaining a full set $\Hs$ that is 
	satisfiable w.r.t.\ $\mathcal{F}$ as described in the proof
	of Thm.\  \ref{thm:sat}. This set $\Hs$ is an extension 
	of $\mathcal{H}$.  The resulting cotree is then returned. 

	If Rule (3a) is satisfied (Line \ref{lin:r3}), 
	Alg.\  \ref{alg:build-cotree} is then called recursively
	to verify whether Rule (3b) is satisfied, and if so, the obtained cotrees
	are joined into a single cotree $(T;t)$. 
	Note, $Q\coloneqq G_{\Rone} / \{C_1,\dots,C_k\}$ is a DAG and 
	there exists a topological order $\pi$ on $Q$. 
	The cotrees are joined to obtain cotree $(T;t)$	from left to right such that
	ordering $\pi$ is preserved, as in the construction described in the proof
	of Thm.\  \ref{thm:sat}. Thus, we obtain a cotree $(T,t)$
  that explains a full
	set $\Hs$ that is satisfiable w.r.t.\ $\mathcal{F}$ and an extension of $\mathcal{H}$.
	 Finally, $(T;t)$ is returned.

	If neither of the rules are satisfied, Thm.\  \ref{thm:sat} implies that $\mathcal{H}$ is
	\emph{not} satisfiable w.r.t.\ $\mathcal{F}$ and the algorithm stops. Hence, a
	cotree is returned that explains the partial set $\mathcal{H}$
	if and only if $\mathcal{H}$ is satisfiable  w.r.t. $\mathcal{F}$.

	For the runtime, observe first that 
	$R_0 \cap F_0 = R_1	\cap F_1 = R_{\Rone} \cap F_{\Rone}$ (Line \ref{lin:intersect-empty}) can
	be computed in $O(m)$ time. Moreover, 
	all digraphs defined in
	Rule (1-3) of Thm.\  \ref{thm:sat}, can be
	constructed in $O(n + m)$ time. Furthermore, each of the following tasks
	can be performed in  $O(n + m)$ time: finding the (strongly)
	connected components of each digraph, building the defined quotient
	graph, and finding the topological order on the quotient graph.
	 Similarly, constructing $V[C_i]$, $\mathcal{H}[C_i]$, and
	$\mathcal{F}[C_i]$ (Line \ref{lin:run-tree1}, \ref{lin:run-tree2}) for each
	(strongly) connected component, can also be done in $O(n	+ m)$ time 
	for all components, by going through every element in $V$, $\mathcal{H}$,
	and $\mathcal{F}$ and assigning each pair to their respective induced subset.

	Thus, every pass of \texttt{BuildCotree} takes $O(n + m)$ time. Since every call
	to \texttt{BuildCotree} adds a vertex to the final constructed
	cotree, and the number of vertices in a tree is bounded by the number $n$ of
	leaves, it follows that \texttt{BuildCotree} can be called  at most $O(n)$
	times. Therefore, we end in an overall running time of $O(m+n(n+m)) =
	O(n^2 + nm)$ for Alg.\  \ref{alg:build-cotree}.
\qed
\end{proof}

Alg.\  \ref{alg:build-cotree} provides a cotree, $(T;t)$, explaining a full
satisfiable set, $\Hs = \{ \Rs_0, \Rs_1, \Rs_{\Rone} \}$, extended from a given
partial set $\mathcal{H}$,
such that $\Hs$ is satisfiable w.r.t.\ a forbidden set $\mathcal{F}$. The
algorithm however, does not directly output $\Hs$. Nevertheless, $\Hs$ can
easily be constructed given $(T;t)$:

\begin{theorem}
	\label{thm:recover-set}
	Let $\mathcal{H} = \{ R_0, R_1, R_{\Rone} \}$ be a partial set that is 
	satisfiable w.r.t\ a forbidden set $\mathcal{F} = \{ F_0, F_1, F_{\Rone} \}$
	and $(T,t)$ be a cotree that explains $\mathcal{H}$. 
	Then, a full satisfiable homology set $\Hs = \{ \Rs_0, \Rs_1, \Rs_{\Rone}	\}$ 
	 that extends 	$\mathcal{H}$ can be constructed in $O(|V|^2)$ time.
\end{theorem}
\begin{proof}
	To recap, the leaf set of $(T=(W,E);t)$ is $L(T)=V$. 
	The full set $\Hs$ can be obtained from $(T;t)$ as follows: For every two 
	leaves $x,y \in V$, where $x$ is left of $y$ in $T$ and $(x,y)\not\in
	R_0\cup R_1 \cup{R_{\Rone}^{sym}}$, 
	we add the unassigned pair $(x,y)$ to $R_0$, $R_1$, or $R_{\Rone}$
	depending on whether $t(\lca(x,y))$ equals $0$, $1$, or $\Rone$
	respectively. If $(x,y)$ was added to $R_0$ (resp.\ $R_1$), then we also add
	$(y,x)$ to $R_0$ (resp.\ $R_1$). 
  Since we address all two 	leaves $x,y \in V$, we obtain 
	$\Rs_0=R_0({\mathcal{T}})$, 	$\Rs_1=R_1({\mathcal{T}})$, 
	and $\Rs_{\Rone}=R_{\Rone}({\mathcal{T}})$. Hence, 
	the resulting	set $\Hs$ is full and explained by
	$(T;t)$. 

	It was shown in \cite{Baruch:88}, that the lowest common ancestor, $\lca(x,y)$, 
	can be accessed in constant time, after an $O(|W|)$ preprocessing step.
	Since we look up the label $t(\lca(x,y))$ for each pair of leaves in $T$, it follows that $\Hs$
	can be constructed in $O(|L(T)|^2) = O(|V|^2)$ time. Hence, $\Hs$ can be
	constructed from $(T,t)$ in $O(n^2)$. 
\qed
\end{proof}


\subsection*{The relative difference of original and recovered full sets}
The relative difference between the true full set $\Hs$ 
on $\Virr$ and the respective recovered full set $\mathcal{H}'$
is given by
\[\frac{|\Rs_0 \Delta R_0'| +|\Rs_1 \Delta R_1'|  + 2|\Rs_{\Rone} \cap \overleftarrow{R'_{\Rone}}| + 2|(\Rs_{\Rone} \setminus \overleftarrow{R'_{\Rone}})\Delta R_{\Rone}'|}{|L|^2-|L|},\]
where $\Delta$ denotes the symmetric set difference.

Here, the term $|L|^2-|L|$ gives the number of all possible elements in a symmetric and irreflexive relation on $L$. 
Thus, the terms $|\Rs_0 \Delta R_0'|/(|L|^2-|L|)$ and $|\Rs_1 \Delta R_1'|/(|L|^2-|L|)$ 
gives the relative difference for the 
symmetric relations $\Rs_0,R_0'$ and $\Rs_1,R_1'$, respectively. For the term  $f=2|\Rs_{\Rone} \cap \overleftarrow{R'_{\Rone}}|/(|L|^2-|L|)$, recall that 
$\Rs_{\Rone}$ and $R_{\Rone}'$ are anti-symmetric relations and observe that
 $(x,y)\in \Rs_{\Rone}, (y,x)\in R_{\Rone}'$ implies that $(x,y)\in \Rs_{\Rone} \cap \overleftarrow{R'_{\Rone}}$. 
The term $\frac{1}{2}(|L|^2-|L|)$ gives the number of all possible elements in an
anti-symmetric irreflexive relation on $L$. Hence, to count the relative differences
of these sets we have to add the term $f$. Furthermore, $(\Rs_{\Rone} \setminus \overleftarrow{R'_{\Rone}})\Delta R_{\Rone}'$
contains all pairs $(x,y) \in \Rs_{\Rone}$ for which there is no pair on $x$ and $y$ contained in $R_{\Rone}'$, 
or {\em vice versa}. By similar arguments as before, we finally have to add the term
$2|(\Rs_{\Rone} \setminus \overleftarrow{R'_{\Rone}})\Delta R_{\Rone}'|/(|L|^2-|L|)$.

%
\bibliographystyle{splncs03}
\bibliography{biblio}

\end{document}